\documentclass[a4paper,USenglish]{lipics}
\usepackage{microtype}
\usepackage{amsmath}
\usepackage{enumerate}
\usepackage{appendix}

\newcommand{\poly}{\mathrm{poly}}
\newcommand{\polylog}{\mathrm{polylog}}
\newcommand{\FRS}{\mathrm{FRS}}
\newcommand{\AFRS}{\epsilon\mathrm{\mbox{-}FRS}}

\theoremstyle{plain}

\begin{document}
\title{Feasible Sampling of Non-strict Turnstile Data Streams}

\author[1]{Neta Barkay}
\author[1]{Ely Porat}
\author[1]{Bar Shalem}
\affil[1]{Department of Computer Science, Bar-Ilan University, Israel\\
  \texttt{\{netabarkay,porately,barshalem\}@gmail.com}}

%
\Copyright[nc-nd]
          {Neta Barkay, Ely Porat and Bar Shalem}


\serieslogo{}
\volumeinfo
  {Billy Editor, Bill Editors}
  {2}
  {Conference title on which this volume is based on}
  {1}
  {1}
  {1}
\EventShortName{}
\DOI{10.4230/LIPIcs.xxx.yyy.p}

\maketitle
\begin{abstract}
We present the first feasible method for sampling a dynamic data stream with deletions,
where the sample consists of pairs $(k,C_k)$ of a value $k$ and its exact total count $C_k$.
Our algorithms are for both Strict Turnstile data streams and the most general Non-strict Turnstile data streams, where each element may have a negative total count.
Our method improves by an order of magnitude the known processing time of each element in the stream, 
which is extremely crucial for data stream applications.
For example, for a sample of size $O(\epsilon^{-2} \log{(1/\delta)})$ in Non-strict streams, our solution requires $O((\log\log(1/\epsilon))^2 + (\log\log(1/\delta)) ^ 2)$ operations per stream element,
whereas the best previous solution requires $O(\epsilon^{-2} \log^2(1/\delta))$ evaluations of a fully independent hash function per element.
Here $1-\delta$ is the success probability and $\epsilon$ is the additive approximation error.

We achieve this improvement by constructing a single data structure from which multiple elements can be extracted with very high success probability. The sample we generate is useful for calculating both forward and inverse distribution statistics,  
within an additive error, with provable guarantees on the success probability.
Furthermore, our algorithms can run on distributed systems and extract statistics on the union or difference between data streams. 
They can be used to calculate the Jaccard similarity coefficient as well.

\end{abstract}

\section{Introduction}

Sampling is a fundamental component in data stream algorithms \cite{muth05},
as a method to keep a synopsis of the data in memory sublinear to the size of the stream.
The sample can be used to calculate stream statistics of interest such as the frequent items in the stream 
(also called \emph{heavy hitters}) or the quantiles, when the stream itself cannot be fully kept in memory.

In the most general data stream model, the data is a series of elements $(x_i,c_i)$ where $x_i$ is an element's value and $c_i$ is a count. A certain value may appear in the stream multiple times with various counts.
Summing all counts of a value $k$ in the stream gives a pair $(k,C_k)$ where $C_k =\sum_{i \colon x_{i}=k}{c_{i}}$ is the total count of $k$.
A value $k$ with $C_k = 0$ is a deleted element, and its effect on stream statistics 
should be as if it had never appeared in the stream. 
Particularly, it must not appear in a sample of the stream obtained by any sampling algorithm.

We denote the function that maps the values to their frequencies by $f$, i.e.
$f(k)=C_k$. 
An \emph{exact sampling algorithm} outputs a sample of the pairs $(k,f(k))$ composed of values and their exact total count.
An \emph{$\epsilon$-approximate sampling algorithm} for $\epsilon \in (0,1)$ outputs a sample of the
pairs $(k,{f'(k)})$, where ${f'(k)} \in [(1-\epsilon)f(k), (1+\epsilon)f(k)]$. 

An $\epsilon$-approximate sample is sufficient for answering \emph{forward distribution} queries, which concern properties of $f$, such as what is $f(k)$.
However, it cannot be used for queries on the \emph{inverse distribution function} $f^{-1}$, defined as
$f^{-1}(C) = \frac{ \left| \{k \colon C_k = C \} \right|}{  \left| \{k \colon C_k \neq 0 \} \right|}$ for $C \neq 0$, i.e. the fraction of distinct values with a total count equal to $C$.
The reason is that an $\epsilon$-approximation of $f$ can result in a significant change to $f^{-1}$.
For example, if an $\alpha \in (0,1)$ fraction of the distinct values have a total count $C$, and in the $\epsilon$-approximate sample all of them have a total count $(1+\epsilon)C$, one might get $f^{-1}(C) = 0$ instead of $\alpha$.
Thus, an exact sample is required in order to approximate the inverse distribution.

The algorithms we present in this paper are exact sampling algorithms for the most general case of streams with deletions. 
Thus, they are useful for calculating both forward and inverse distribution statistics.
We describe applications which require exact sampling next.

Data stream algorithms are often used in network traffic analysis.
They are placed in DSMSs -  Data Stream Management Systems.
Applications of dynamic exact sampling in these environments include
detecting malicious IP traffic in the
network. Karamcheti et al.\ \cite{karamcheti05} showed
that inverse distribution statistics
can be used for earlier detection of content similarity, an
indicator of malicious traffic.
Inverse distribution queries can also be used for detecting denial-of-service (DoS) attacks, specifically SYN floods,
which are characterized as flows with a single packet (often called \emph{rare flows}).

Exact dynamic sampling is beneficial in geometric data streams as
well \cite{frahling05,frahling08}. In these streams, the items
represent points from a discrete $d$-dimensional geometric space ${\{1,\ldots,\Delta\}}^{d}$.
Our algorithms can also run on data streams of this type.

\paragraph*{Previous Work}

Most previous work that has been done on sampling of dynamic data streams that support
deletions was limited to approximating
the forward distribution \cite{cohen12, gemulla07, mone10}. 
Works on the inverse distribution include a restricted model with total counts
of 0/1 for each value \cite{gemulla06, tao07}, and
minwise-hashing, which samples uniformly the set of items but does
not support deletions \cite{datar03}.
The work \cite{gibbons01} supports only a few deletions.

Inverse distribution queries in streams with multiple deletions were supported in a work by Frahling et al.\ \cite{frahling05,frahling08},
who developed a solution for Strict Turnstile data streams and used it
in geometric applications. 
Cormode et al.\ \cite{cormode05} developed a solution for both the
Strict Turnstile and the Non-strict Turnstile streams. However, they
did not analyze the required randomness or the algorithm's error probability
in the Non-strict model.
Jowhari et al.\ \cite{jowhari11} studied $L_p$ samplers \cite{mone10} and built an $L_0$ sampler for Non-strict Turnstile streams. 

\paragraph*{Our Results}

Previous works \cite{cormode05,frahling05,frahling08,jowhari11} constructed data
structures for sampling only a single element. 
In order to use their structures for applications that require a sample of size $K$,
one has to use $O(K)$ independent instances of their structure.
The obtained sample holds elements chosen independently, and it might contain
duplicates. 

Running the sampling procedure $O(K)$ times in parallel and inserting each element in the stream as an input to $O(K)$ instances of the data structure, results in an enormous process time.
Typical stream queries require a sample of size $K = \Omega(\frac{1}{\epsilon^2}\log{\frac{1}{\delta}})$ where the results are $\epsilon$ approximated, and $1-\delta$ is the success probability of the process.
Thus, the number of operations required to process each element in the stream is multiplied by many orders of magnitude. For typical values such as $\epsilon = 10^{-2}$, $\delta = 10^{-6}$ the number of operations for each element in the stream is multiplied by about $200{,}000$.
The structures of \cite{cormode05,frahling05,frahling08,jowhari11} cannot be used for obtaining a $K$ size sample due to this unfeasible process load.
We present algorithms that can accomplish the task.

Our contributions are as follows. 
\begin{itemize}

	\item 
	We construct a data structure that can extract a sample of size $K$, whereas previous works returned only a single sampled element.
Using a single data structure reduces significantly the process time and the required randomness.
Thus, our algorithms are feasible for data stream applications, in which fast processing is crucial. 
This optimization enables applications that were previously limited to gross approximations of order $\epsilon = 10^{-2}$ to obtain a much more accurate approximation of order $\epsilon = 10^{-6}$ in feasible time. 

\item We present solutions for both Strict Turnstile data streams and the most general Non-strict Turnstile
data streams. We are the first to provide algorithms
with proved success probability in the Non-Strict Turnstile
model.
For this model we develop a structure called the Non-strict Bin Sketch.

\item We provide more efficient algorithms in terms of the randomness required. Our algorithms do not require fully independent or min-wise independent hash functions or PRGs.

\item 
We introduce the use of $\Theta(\log{\frac{1}{\delta}})$-wise independent hash functions to generate a sample with $1-\delta$ success probability for any $\delta >0$.
Our method outperforms the traditional approach of increasing the success probability from a constant to $1-\delta$ by $\Theta(\log\frac{1}{\delta})$ repetitions. 
We utilize a method of fast evaluation of hash functions which reduces our processing time to $O((\log\log\frac{1}{\delta})^2)$, while the traditional approach requires $O(\log\frac{1}{\delta})$ time.

\end{itemize}

A comparison of our algorithms to previous work is presented in Table~\ref{table:previouswork}. 
We introduce two algorithms, denoted $\FRS$ 
(Full Recovery Structure)
 and $\AFRS$ according to the recovery structure used.
The performance of our algorithms is summarized in the table as well as in Theorems~\ref{theorem1} and~\ref{theorem2}. Our algorithms improve the update and the sample extraction times.

\begin{table}[h] 

\caption{Comparison of our sampling algorithms for streams with deletions to previous work.}
{
\renewcommand{\tabcolsep}{3pt}
\begin{tabular}{|p{0.9cm}|p{3.7cm}|p{2.5cm}|p{2.5cm}|p{2.3cm}|p{0.6cm}|}
\hline
\centering{
Alg.} & \centering{Memory size} & \centering{Update time given hash oracle} & \centering{Overall update time per element}& \centering{Sample extraction time}& Model \\

\hline

\cite{frahling05, frahling08} & $O(K \log^2(\frac{mr}{\delta}))$ & $O(K \log{m})$ & $O(K \cdot \polylog{m})$ & $O(K)$ & S\\

\cite{frahling08} & $O(K \log{m} \frac{1}{\epsilon} \log\frac{1}{\delta} \log(\frac{mr}{\epsilon}))  $ & 
$O(K\frac{1}{\epsilon} \log\frac{1}{\delta} \log{m}) $ & $O(K \frac{1}{\epsilon} \log\frac{1}{\delta} \log{m})$ &
$O(K \frac{1}{\epsilon} \log{\frac{1}{\delta}})$& S \\

\cite{cormode05} & $O(K\log^2{m} \log{(mr)}) $ & $O(K\log{m})$ & $O(K\log{m})$ & $O(K\log^2{m})$ &S\\

\cite{cormode05} & $O(K\log{m} \log{\frac{1}{\delta}} \log{(mr)}) $ & $O(K \log{\frac{1}{\delta}})$ & $-$ & $O(K \log{m} \log{\frac{1}{\delta}})$ & NS \\ 

\cite{jowhari11} & $O(K \log^2{m} \log\frac{1}{\delta})$ & $O(K)$ & $O(K\log{m})$ & $O(K\log{m}\log{\frac{1}{\delta}})$ & NS \\

\hline

$\FRS$ & $O(K \log{m} \log\frac{K}{\delta} \log{(mr)})$ & $O(\log\frac{K}{\delta})$ & $O(\log\frac{K}{\delta})$ & $O(K \log\frac{K}{\delta})$ & S/NS \\

$\AFRS$ & $O(K \log{m}  \log{(mr)})$ & $O(1)$ & $O({(\log\log\frac{K}{\delta})}^2)$ & $O(K)$  & S\\

$\AFRS$ & $O(K \log{m}  \log{(\frac{mr}{\delta})})$ & $O(1)$ & $O({(\log\log\frac{K}{\delta})}^2)$ & $O(K)$ & NS\\

\hline

\multicolumn{6}{l}{\begin{minipage}{1\textwidth}~\\
Notes:
$\FRS$ and $\AFRS$ support sample sizes $K = \Omega(\log\frac{1}{\delta})$ and $K = \Omega(\frac{1}{\epsilon}\log\frac{1}{\delta})$ respectively.
Memory size is given in bits.
Update time is the number of operations per element.
Sample extraction time is for a $K$ size sample.
S and NS denote the Strict and Non-strict models respectively.
In \cite{cormode05} S is the deterministic approach, NS is the probabilistic approach, and the hash function is fully independent.
In \cite{jowhari11} the extraction time is under the assumption of sparse recovery in linear time.

\end{minipage}
}
\end{tabular}
}
\label{table:previouswork}
\end{table}

\section{Preliminaries}

\subsection{Data Stream Model}

Our input is a stream in the \emph{Turnstile} data stream model \cite{muth05}.
The Turnstile data stream consists of $N$ pairs $(x_{i},c_{i})$
where $x_{i}$ is the element's value, and $c_{i}$ is its count.
The elements $x_i$ are taken from a fixed universe $U=[m]$ (where
$[m]=\{0,\ldots,m-1\}$). The counts $c_i$ are taken from the range $[-r,r]$.
Let $t$ be the time we process the $t$'th pair in the data stream.
We define the total count of value $k$ at time $t$ to be $C_{k}(t)=\sum_{i \le t \colon x_{i}=k}{c_{i}}$.
We assume that $\forall t,k, C_{k}(t)\in[-r,r]$. In the \emph{Strict
Turnstile} model, $C_{k}(t)\geq0$ at all times $t$. In the \emph{Non-strict
Turnstile} model, $C_{k}(t)$ may obtain negative values.

A sample $S$
drawn at time $t$ is a subset of ${\{(k,C_{k}(t)) \colon C_k(t) \neq 0\}}$.
Note that $C_{k}(t)$ is the exact total count at time $t$. To simplify the notation we consider sampling only at the end of the process, and denote $C_{k}=C_{k}(N)$.

\subsection{Problem Definition}

Given a data stream of $N$ pairs $(x_i, c_i)$,
which is either a Strict or Non-strict Turnstile data stream, assume 
there is an application that needs to perform some queries on the stream such as calculate the inverse distribution.
This application allows an $\epsilon \in (0,1)$ additive approximation to the answers, and a $1-\delta$ for $\delta \in (0,1)$ success probability of the process.
The application predefines the size of the sample it requires to be $K$, where
$K$ might depend on $\epsilon$ and $\delta$.
The input to our sampling algorithm is the data stream, $K$ and $\delta$.

Let $D = {\{ (k, C_{k}) \colon C_k \neq 0\} }$ be the set of all pairs of values and their total counts in the stream at the time we sample.
The output of our sampling algorithm is a sample $S \subseteq D$ of size $\left|S\right| = \Theta(K)$,
generated with probability $1-\delta$.
Note that the size of the sample returned is of order $K$ and not precisely $K$.

Applications typically require a sample of size $K=\Omega(\frac{1}{{\epsilon}^{2}}\log{\frac{1}{\delta}})$ for an $\epsilon$ approximation with $1-\delta$ success probability. However, our algorithms $\FRS$ and $\AFRS$, support even smaller sample sizes, $\Omega(\log{\frac{1}{\delta}})$ and $\Omega(\frac{1}{\epsilon}\log{\frac{1}{\delta}})$ respectively. 

We define the following two ``flavors'' of samples.

\begin{definition}
A \emph{$t$-wise independent sample} is a random set $S \subseteq D$ in which each subset of $t$ distinct elements in $D$ has equal probability to be in $S$.
\end{definition}

\begin{definition}
Let $X \subseteq D$ be a sample obtained by a $t$-wise independent sampling algorithm, and let $\epsilon \in (0,1)$. 
A subset $S \subseteq X$ of size $(1-\epsilon)\left| X \right| \leq \left| S \right| \leq \left| X \right|$ is a \emph{$(t,\epsilon)$-partial sample}.
\end{definition}

Our $\FRS$ algorithm returns a $t$-wise independent sample, where $t=\Omega(\log\frac{1}{\delta})$. 
Our  $\AFRS$ algorithm returns a $(t,\epsilon)$-partial sample.
This means there is a fractional bias of at most $\epsilon$ in the sample returned by $\AFRS$.

The key insight is that a $t$-wise independent sample for $t = \Omega{(\log\frac{1}{\delta})}$ guarantees the same approximation as an independent sample.
For example, a sample of size
$K=\Omega(\frac{1}{{\epsilon}^{2}}\log{\frac{1}{\delta}})$ enables the approximation of the inverse distribution
queries and the Jaccard similarity coefficient up to an additive error of $\epsilon$ with $1-\delta$
success probability.
The $(t,\epsilon)$-partial sample can be used for the same stream statistics because it only adds an error of at most $\epsilon$ to the approximation.
This is demonstrated in Sect.~\ref{sec:apps}.

\subsection{Hash Function Techniques}

Throughout the paper we make extensive use of $t$-wise independent hash functions for $t= \Theta(\log\frac{1}{\delta})$ and $t= \Theta(\log\frac{K}{\delta})$.
We use the following techniques in our analysis.
\label{sec:highmoments}
For bounding the error probability we use the Moment Inequality for high moments and the estimation of \cite{indyk01}
(see Appendix \ref{appendix:chebyshev}).
For hash evaluations we use the multipoint evaluation algorithm of a polynomial of degree less than $t$ on $t$ points in $O(t\log^{2}{t})$ operations \cite{gathen99}.
Thus, evaluation of a $t$-wise independent hash function takes $O(\log^{2}{t})$ amortized time
per element by batching $t$ elements together. 
This is the time we use in our analysis whenever we evaluate these hash functions.

\section{Algorithm Overview}

In this section we provide an overall view of our sampling algorithm (see Fig.~\ref{fig:pseudocode}).

\begin{figure}[h]
\begin{centering}
\framebox[1\width]{
\begin{tabular}{p{175pt} p{10pt} p{175pt}}
{ Data structure update: For each pair $\rho=(x_{i},c_{i})$ in the
stream: 
\begin{enumerate}[1.]
\item Map $\rho$ to its corresponding level $l$. 
\item Update the recovery structure, $\FRS$ ($\AFRS$), in $l$:
	\begin{enumerate}[2.1.]
	\item For each array in $\FRS$ ($\AFRS$):
	\item Insert $\rho$ to $BS$ in the corresponding cell.
	\end{enumerate}
\item Update the $L_0$ estimation structure.

\end{enumerate}
}  &  & {Extracting a $K$ size sample:
\begin{enumerate}[1.]
\item Get $\tilde{L}_{0}$ from the estimation structure. 
\item Determine the level $l^{*}$ to be used for recovery. 
\item Recover the elements from $\FRS$ ($\AFRS$) in level $l^{*}$. 
\item Return the $\Theta(K)$ size sample. 
\end{enumerate}
}

\end{tabular}}
\end{centering} 
\caption{Pseudo codes of data structure update and extracting a $K$ size sample.}
\label{fig:pseudocode}
\end{figure}

In the first phase of our sampling algorithms the elements in the data stream are mapped to levels,
in a similar way to the method in \cite{cormode05}. 
Each element is mapped to a single level. The number of elements mapped to each level decreases exponentially.
Thus, we can draw a sample of the required size regardless of the number of distinct elements in the stream.

In each level we store a recovery structure for $K$ elements, which is the core of our algorithm. We present two alternative
structures, $\FRS$ and $\AFRS$, with trade-offs between space, update time and sample extraction time. 
Each structure consists of several arrays with a Bin Sketch ($BS$) in each cell of the arrays.
Assume a structure ($\FRS$ or $\AFRS$) contains the set of elements $X$, where \mbox{$K \leq \left| X \right| \leq \tilde{K}$} and $\tilde{K} = \Theta(K)$.
$\FRS$ enables to extract all the elements it contains, returning a $t$-wise independent sample, where $t = \Theta(\log\frac{1}{\delta})$.
$\AFRS$ enables to extract at least $(1-\epsilon)\left|X\right|$ of the $\left|X\right|$ elements it contains, returning a $(t,\epsilon)$-partial sample.
The problem of recovering the $\left|X\right|$ elements is similar to the sparse recovery problem \cite{gilbert10,porat07}, however in our case there is no tail noise, and we limit the amount of randomness.

The sample $S$ is the set of elements extracted from $\FRS$ (or $\AFRS$) in a single level.
In order to select the level we use a separate data structure that returns $\tilde{L}_0$, an estimation of the number of distinct elements in the stream.
This data structure is updated as each stream element arrives, in parallel to the process described above.

Extracting a sample is performed as follows. First we query the data structure of $\tilde{L}_0$. Then we select a level $l^{*}$ that should have $\Theta(K)$ elements with probability $1-\delta$. We recover
the $X$ elements from that level, or at least $(1-\epsilon)\left|X\right|$ of them, depending on the recovery structure, with probability
$1-\delta$. The elements recovered are the final output sample $S$.

\section{Sampling Algorithms} \label{sec:work}

\subsection{Bin Sketch - Data Structure Building Block}

\label{3 counters}

In this section we present the \emph{Bin Sketch ($BS$)}, a tool used as a building block in our data structure.
Given a data stream of elements ${\{(x_i,c_i)\}}_{i \in [N]}$, the input to $BS$ is a substream ${\{(x_i,c_i)\}}_{i \in I}$ where $I \subseteq [N]$.
$BS$ maintains a sketch of the substream. 
Its role is to identify if the substream contains a single value, and if so, to retrieve the value and its total count. 
Note that a single value can be obtained from a long stream of multiple elements if all values but one have zero total count at the end.

\paragraph*{Strict Bin Sketch ($BS_{s}$)}

We describe the Bin Sketch for the Strict Turnstile
model, previously used in \cite{frahling08,ganuly07}.
Given a stream of pairs $\{(x_{i},c_{i})\}_{i \in I}$, the \emph{Strict
Bin Sketch ($BS_{s}$)} consists of three counters $X$, $Y$
and $Z$:
$
X=\sum_{i \in I}{c_{i}},\ Y=\sum_{i \in I}{c_{i}x_{i}},\ Z=\sum_{i \in I}{c_{i}x_{i}^{2}}$.

$BS_s$ properties are summarized as follows:
The space usage of $BS_{s}$ is $O(\log(mr))$ bits. 
If $BS_{s}$ holds a single element, then we can detect and recover it.
$BS_{s}$ holds a single element $\Leftrightarrow$ $X\neq0$, $Y\neq0$, $Z\neq0$ and $XZ=Y^{2}$. 
The recovered element is $(k,C_{k})$, where $k=Y/X$ and
$C_{k}=X$. $BS_{s}$ is empty $\Leftrightarrow$ $X=0$.
For proof see \cite{ganuly07}.

\paragraph*{Non-strict Bin Sketch ($BS_{ns}$)}

We provide a generalization to Bin Sketch and adjust it to the Non-strict
Turnstile model.
If $BS_s$ contains two elements, then $XZ\neq Y^{2}$.
Thus, there is a distinction between $BS_s$ with two elements and $BS_s$ with a single element.
However, we cannot always distinguish between three or more elements and a single element.

A previous attempt to solve the problem used a deterministic collision
detection structure \cite{cormode05}. This structure
held the binary representation of the elements. However, this representation
falsely identifies multiple elements as a single element on some inputs\footnote{
For example, for every set of 4 pairs $\{(2k,1),(2k+1,-1),(2k+2,-1),(2k+3,1)\}$,
the whole structure is zeroed and any additional pair $(k^{'},C_{k}^{'})$
will be identified as a single element.}.

In order to solve the problem we use a new counter defined as follows.

\begin{lemma} \label{lemma:countersq} 
Let $\mathcal{H}=\{h \colon \left[m\right]\rightarrow\left[q\right]\}$
be a $t$-wise independent family of hash functions, and let $h$
be a hash function drawn randomly from $\mathcal{H}$.
Let $T=\sum_{k}{C_{k}h(k)}$ be a counter.
Assume $T$ is a sum of at most $t-1$ elements.
Then for every element $(k',C_{k}')$, where $k'\in\left[m\right]$, if $T$ is not the sum of the single element $(k',C_{k}')$, 
then $\Pr_{h \in \mathcal{H}}[T = C_{k}'h(k')] \leq 1/q$.
\end{lemma}

\begin{proof} We subtract $C_{k}'h(k')$ from $T$ and obtain $T'=T-C_{k}'h(k')$.
If $T$ is not the sum of the single element $(k',C_{k}')$, there
are between $1$ and $t$ elements in $T'$. The hashes of those elements
are independent because $h$ is $t$-wise independent.
We therefore have a linear combination of independent uniformly distributed
elements in the range $[q]$, thus 
$\Pr_{h \in \mathcal{H}}[T'=0] \leq 1/q$.
\end{proof}

The \emph{Non-strict Bin Sketch ($BS_{ns}$)} consists of four
counters: $X$,$Y$,$Z$ and an additional counter $T$, as defined
in Lemma \ref{lemma:countersq}. 
The space of $BS_{ns}$ is $O(\log{(mrq)})$ bits.
The time to insert an element is $O(\log^2{t})$ since we evaluate a $t$-wise independent hash function.
There are other ways to maintain a sketch such as keeping a fingerprint. However, fingerprint update takes $O(\log{m})$ time, which depends on the size of the universe, while the update time of $BS_{ns}$ depends on $t$ which is later set to $t = \Theta(\log\frac{K}{\delta})$.

\begin{corollary} Three or more elements in $BS_{ns}$ may be falsely identified 
as a single element. The error probability is bounded
by $1/q$, when there are at most $t-1$ elements in the bin. \end{corollary}

$BS$ is placed in each cell, which we refer to as \emph{bin}, in each of the arrays of our data structure. Its role is
to identify a \emph{collision}, which occurs when more than one element
is in the bin. If no collision occurred, $BS$ returns the
single element in the bin.
$BS$ supports additions and deletions and is oblivious of the order of their occurrences in the stream, which makes it a \emph{strongly history independent} data structure \cite{micciancio97,naor01}. Its use in our sampling data structure makes the whole sampling structure strongly history independent.

\subsection{$L_{0}$ Estimation}

The number of distinct values in a stream with deletions is known as the \emph{Hamming norm} $L_{0}=\left|\{k\in[m] \colon C_{k}\neq0\}\right|$
\cite{cormode03}.
$L_{0}$ Estimation is used by our algorithm for choosing the level
$l^{*}$ from which we extract the elements that are our sample.
We use the structure of Kane et al.\ \cite{kane10}, which
provides an $\epsilon$-approximation $\tilde{L}_{0}$
with $2/3$ success probability for both Strict and Non-strict
streams. 
It uses space of $O(\frac{1}{\epsilon^{2}}\log{m}(\log\frac{1}{\epsilon}+\log\log(r)))$
bits and $O(1)$ update and report times.

Our algorithm requires only a constant approximation to $L_0$. Hence the space is $O(\log{m} \cdot \log\log(r)))$ bits.
However, we require $1-\delta$ success probability for any given $\delta > 0$ and not only $2/3$.
For estimation of $L_0$, where $L_0 = \Omega({\frac{1}{\epsilon^2}\log\frac{1}{\delta})}$ we use methods similar to those in the rest of the paper. We keep $\tau = \Theta(\log\frac{1}{\delta})$ instances of Kane et al. structure, 
and use one $\Theta(\log\frac{1}{\delta})$ independent hash function to map each stream element to 
one of the instances.
With constant high probability all instances have approximately the same number of elements.
$\tilde{L}_0$ is the median of the estimations obtained from all instances multiplied by $\tau$, and is a constant approximation to $L_0$ with probability $1-\delta$.

Thus we obtain an $L_0$ estimation algorithm with $O(\log{m} \log\log(r) \log\frac{1}{\delta})$ bits of space, $O((\log\log\frac{1}{\delta})^2)$ update time, since one $O(\log\frac{1}{\delta})$ independent hash function is evaluated, $O(\log\frac{1}{\delta})$ reporting time, and $1-\delta$ success probability.
These requirements are lower than their corresponding requirements in the other phases of our algorithm.

\subsection{Mapping to Levels}

The first phase of our algorithm is mapping the elements in the data stream to levels.
We use a family of $t$-wise 
independent hash functions $\mathcal{H}=\{h \colon [m]\rightarrow[M]\}$ for $t=\Theta(\log{\frac{1}{\delta}})$.
We select $h\in\mathcal{H}$ randomly and use it to map the elements
to $L=\log_{1/\lambda}{M}$ levels, for $\lambda \in (0,1)$. Typical values are $\lambda=0.5$ and
$M=2m$. The
mapping is performed using a set of hash functions $h_{l}(x)=\left\lfloor {\frac{h(x)}{{\lambda}^{l}M}}\right\rfloor $,
for $l\in[L]$. An element $x$ is mapped to level $l$ $\Leftrightarrow$
$(h_{l}(x)=0\wedge h_{l+1}(x)\neq0)$. Note that each element is mapped
to a single level.

Using this mapping, the set of elements mapped to each level is $t$-wise independent. It follows that in order to obtain a $t$-wise independent sample, we can extract all elements from any level we choose. However, we must select
the level independently of the elements that were mapped to it. If
any event that depends on the specific mapping influences the level selection,
the sample becomes biased. Biased samples appeared in some previous works.

In order to choose the level regardless of the mapping, we use the number of distinct elements
$L_{0}$. 
We obtain an estimation $\tilde{L}_{0}$ from the $L_{0}$ estimation structure, where $L_{0}\leq\tilde{L}_{0}\leq\alpha L_{0}$ for $\alpha>1$ with $1-\delta$ probability, and choose the level where $K$ elements are expected.

\begin{lemma} \label{lemma:mapping} 
Let the elements be mapped to levels
using a hash function $h$ selected randomly from a $t$-wise
independent hash family $\mathcal{H}$ for $t=\Omega(\log{\frac{1}{\delta}})$. Assume there is an estimation
$L_{0}\leq\tilde{L}_{0}\leq\alpha L_{0}$ for $\alpha>1$, and $K = \Omega({\log\frac{1}{\delta}})$. Then the
level $l^{*}$ for which $\frac{1}{\alpha}\tilde{L}_{0}\lambda^{l^{*}+1}(1-\lambda)<2K\leq\frac{1}{\alpha}\tilde{L}_{0}\lambda^{l^{*}}(1-\lambda)$
has $K$ to $(\frac{2\alpha}{\lambda}+1)K$ elements with probability at least
$1-\delta$.
\end{lemma}

\begin{proof} See Appendix \ref{appendix:proofs}.
\end{proof}

Let $X$ be the set of elements in level $l^{*}$, from which we choose to extract the sample. 
We denote $\tilde{K} = (\frac{2\alpha}{\lambda}+1)K$ the maximal number of
elements in level $l^{*}$. With probability at least $1-\delta$, $K \leq \left| X \right| \leq \tilde{K}$.
For typical values $\alpha = 1.5$ and $\lambda = 0.5$, $K \leq \left| X \right| \leq 7K$.

\subsection{Full Recovery Data Structure ($\FRS$)}

In this section we present the Full Recovery Data Structure ($\FRS$) that is placed in each of the levels. 
Since the recovery structure is the core of our sampling algorithm, we refer to the entire algorithm with $\FRS$ in each level as the $\FRS$ algorithm.
The sampling algorithm can be summarized to the following theorem.

\begin{theorem}	\label{theorem1}
Given a required sample size $K = \Omega(\log\frac{1}{\delta})$ and $\delta \in (0,1)$, $\FRS$ sampling algorithm generates a $\Theta(\log\frac{1}{\delta})$-wise independent sample $S$ with $1-\delta$ success probability. The sample size is $K \leq \left| S \right| \leq \tilde{K}$, where $\tilde{K} = \Theta(K)$.
For both Strict and Non-strict data streams $\FRS$ uses $O(K \log\frac{K}{\delta} \log{(mr)} \log{(m)})$ bits of space,
$O(\log\frac{K}{\delta})$ update time per element, $O(\log\frac{K}{\delta} \log{m})$ random bits and $O(K \log\frac{K}{\delta})$ time to extract the sample $S$.
\end{theorem}

$\FRS$ is inspired by Count Sketch \cite{charikar04}. It is composed of $\tau=O(\log{\frac{K}{\delta}})$ arrays
of size $s=O(K)$. We refer to the cells in the array as \emph{bins}.
Each input element is mapped to one bin in each of the $\tau$ arrays.
In each bin there is an instance of $BS_{s}$. We use $\tau$
hash functions drawn randomly and independently from a
pairwise independent family $\mathcal{H}=\{h \colon [m]\rightarrow[s]\}$. 
The same set of hash functions can be used for the instances of $\FRS$ in all levels.
The mapping is performed as follows.

Let $B[a,b]$ for $a\in[\tau]$ and $b\in[s]$, be the $b$'th bin in
the $a$'th array. Let the hash functions be $h_{1}\ldots h_{\tau}$.
Then $(x_{i},c_{i})$ is mapped to $B[a,h_{a}(x_{i})]$ for every
$a\in[\tau]$.
We say that two elements \emph{collide} if they are mapped to the
same bin. 

\begin{lemma} \label{Collisions in arrays} Let $\FRS$ with at most $\tilde{K}$
elements have $\tau=\log{\frac{\tilde{K}}{\delta}}$ arrays of size $s=2\tilde{K}$.
Then with probability at least $1-\delta$ for each element there is a bin
in which it does not collide with any other element. \end{lemma}

\begin{proof} See Appendix \ref{appendix:proofs}. \end{proof}

\begin{corollary} \label{Alg1Strict} In the Strict Turnstile model all
elements in $\FRS$ can be identified and recovered with probability at least $1-\delta$. \end{corollary}

For recovery, we scan all bins in all arrays in $\FRS$ and use $BS_{s}$ to extract elements
from all the bins that contain a single element. According to Lemma
\ref{Collisions in arrays}, all elements in $\FRS$ can be identified
with success probability $1-\delta$. We verify success by removing all
the elements we found and scanning the arrays an additional time to
validate that they are all empty.
Removing an element $(x,c)$ is performed by inserting $(x,-c)$ to the corresponding bins.

\subsubsection{Non-strict $\FRS$}

We now present the generalization of $\FRS$ to Non-strict streams. Once
again we use $BS_{s}$, but we add to our sample only elements that are consistently
identified in multiple arrays.

\begin{lemma} \label{lemma:alg1NS} Let $\FRS$ with at most $\tilde{K}$
elements have $\tau=5\log{\frac{\tilde{K}}{\delta}}$ arrays of size
$s=8\tilde{K}$. In the Non-strict data stream model, all elements
inserted to $\FRS$ and only those elements are added to the sample
with probability at least $1-\delta$. \end{lemma}

\begin{proof} 
We extract a set $A$ of candidate elements from all $BS_{s}$s that seem to have a single element in the first $\log\frac{\tilde{K}}{\delta}$ arrays in $\FRS$. 
$A$ contains existing elements, that were inserted to $\FRS$, and falsely detected elements that are a result of a collision.
$\left| A \right| \leq \tilde{K} \log\frac{\tilde{K}}{\delta}$.
It follows from Lemma \ref{Collisions in arrays} and Corollary \ref{Alg1Strict} that all of the existing elements can be recovered from the first $\log\frac{\tilde{K}}{\delta}$ arrays with probability $1-\delta/2$ (increasing the arrays size reduces the probability of a collision).
Hence $A$ contains all existing elements with probability $1-\delta/2$.

Next we insert to our output sample all candidates that we detect in at least half of their bins in the $\tau' = 4\log{\frac{\tilde{K}}{\delta}}$ remaining arrays of $\FRS$.
There are two types of possible errors: not reporting an existing element (false negative) and reporting a falsely detected element (false positive).

\textbf{False negative:}
We show that with high probability existing elements
are isolated in at least half of the their bins in the $\tau'$ arrays.
Let $\mathcal{C}_{k}^{a}$ be the event that element $k$
collides with another element in array $a$. 
$\Pr[\mathcal{C}_{k}^{a}]< \tilde{K}\cdot\frac{1}{s}=\frac{1}{8}$.
Let $\mathcal{C}_{k}$ be the event that element $k$ collides with
another element in at least half of the $\tau'$ arrays. The hash functions
of the different arrays are independent and therefore: 
$\Pr[\mathcal{C}_{k}] \leq {\tau' \choose \tau'/2}\Pr[\mathcal{C}_{k}^{a}]^{\tau'/2}<\left(\frac{\delta}{\tilde{K}}\right)^{2},$
where ${\tau' \choose \tau'/2}<2^{\tau'}$ is used.
Let $\mathcal{C}$ be the event that there is an existing element that collides
with another element in at least half of the $\tau'$ arrays. $\Pr[\mathcal{C}] \leq \tilde{K}\cdot\Pr[\mathcal{C}_{k}]<\tilde{K}\left(\frac{\delta}{\tilde{K}}\right)^{2} < \frac{\delta}{4}.$

\textbf{False positive:}
If a falsely detected element from $A$ is added to the sample then there is a collision in at least half of its bins in the $\tau'$ arrays.
Let $\mathcal{E}_{b}^{a}$ be the event that there is an element
in bin $b$ of array $a$. 
$\Pr[{\mathcal{E}}_{b}^{a}] \leq {\tilde{K} \choose 1}\left(\frac{1}{s}\right)=\frac{1}{8}$.
Let $\mathcal{E}_{k}$ be the event that there are elements
in the bins corresponding to a falsely detected element $k$ in at least
half of the $\tau'$ arrays. $
\Pr[\mathcal{E}_{k}] \leq {\tau' \choose \tau'/2}\Pr[{\mathcal{E}}_{b}^{a}]^{\tau'/2}<2^{\tau'}\left(2^{-3}\right)^{\tau'/2}
=2^{-0.5\tau'}
=\left(\frac{\delta}{\tilde{K}}\right)^{2}$.
Let $\mathcal{E}$ be the event that there is an element from $A$
that was falsely identified in at least half of its bins. Using the
union bound we get:
$\Pr[{\mathcal{E}}] \leq \left|A\right| \cdot\Pr[{\mathcal{E}}_{k}] < \tilde{K} \log\frac{\tilde{K}}{\delta} \cdot \left(\frac{\delta}{\tilde{K}}\right)^{2} < \frac{\delta}{4}$.

We conclude that the probability of a mistake is bounded by: 
$\delta/2 + \Pr[\mathcal{C}]+\Pr[\mathcal{E}]<\delta$.
\end{proof}

\subsection{$\epsilon$-Full Recovery Data Structure ($\AFRS$)}

In this section we present the $\epsilon$-Full Recovery Data Structure ($\AFRS$) that enables to recover almost all elements inserted to it. We refer to the entire algorithm with $\AFRS$ placed in each of the levels as $\AFRS$ algorithm. The sampling algorithm can be summarized to the following theorem.

\begin{theorem}	\label{theorem2}
Given a required sample size $K = \Omega(\frac{1}{\epsilon} \log{\frac{1}{\delta}})$, for $\delta \in (0,1)$ and \mbox{$\epsilon \in (0,1)$},
$\AFRS$ sampling algorithm generates a $(t,\epsilon)$-partial sample $S$ for $t = \Theta(\log\frac{1}{\delta})$ with $1-\delta$ success probability. The sample size is $(1-\epsilon)K \leq \left|S\right| \leq \tilde{K}$, where $\tilde{K} = \Theta(K)$.
For both Strict and Non-strict data streams $\AFRS$ requires $O((\log\log\frac{K}{\delta})^2)$ update time per element, $O(\log\frac{K}{\delta} \log{m})$ random bits and $O(K)$ time to extract the sample $S$.
The space is $O(K \log{(mr)} \log{(m)})$ bits for Strict data streams and $O(K \log{(\frac{mr}{\delta})} \log{(m)})$ bits for Non-strict streams.
\end{theorem}

$\AFRS$ is composed of $\tau=2$ arrays
of size $s=O(K)$. As in $\FRS$,
each input element is mapped to one bin in each of the arrays.
In each bin of each array we keep an instance of $BS_{s}$ or $BS_{ns}$ according to the input data stream. 
The mapping is performed using two hash functions drawn randomly and independently from a $t$-wise independent family $\mathcal{H}=\{h \colon [m]\rightarrow[s]\}$ for $t=\Theta(\log{\frac{K}{\delta}})$.

Let $X$ be the set of elements in $\AFRS$, $\left|X\right|\leq \tilde{K}$.
A \emph{fail set} $F\subseteq X$ is a set of $f$
elements, such that each element in the set collides with other
elements from the set in both its bins.
The elements in a fail set $F$ cannot be extracted from $\AFRS$.
Analyzing the existence of a fail set is similar to analyzing failure in a cuckoo hashing \cite{pagh01} insertion.
We bound the probability that there is a fail set of size $f$ using the following (revised) lemma of Pagh and Pagh \cite{pagh08}.
\begin{lemma}[\cite{pagh08},Lemma3.4]
For two functions $i_1,i_2 \colon U \rightarrow [R]$ and a set $S \subseteq U$, let $G(i_1,i_2,S)=(A,B,E)$ be the bipartite graph that has left vertex set $A = \{a_1,\ldots,a_R\}$, right vertex set $B = \{b_1,\ldots,b_R\}$ and edge set $E = \{e_x \mid x \in S \}$, where $e_x = (a_{i_1(x)},b_{i_2(x)})$.

For each set $S$ of size $n$, and for $i_1,i_2:U \rightarrow [4n]$ chosen at random from a family that is $t$-wise independent on $S$, $t \geq 32$, the probability that the fail set $F$ of the graph $G(i_1,i_2,S)$ has size at least $t$ is $n/{2^{\Omega(t)}}$.
\end{lemma}

\begin{corollary}
Let $\AFRS$ with at most $\tilde{K}$ elements have 2 arrays of size $s = 4\tilde{K}$.
Let the mapping be performed by two $t$-wise independent hash functions for $t=c \log\frac{\tilde{K}}{\delta}$, constant $c$ and $t \geq 32$.
The probability that there is a fail set of size at least $t$ is bounded by $\delta$.
\end{corollary}

\begin{proof}
The more elements in $\AFRS$, the higher the probability that there is a fail set of some fixed predefined size.
The probability is $\tilde{K} / {2^{c' c \log{\frac{\tilde{K}}{\delta}}}} \leq \delta$ for some constants $c$, $c'$.
\end{proof}

The following algorithm identifies all elements in $\AFRS$ that do not belong to a fail set.

\begin{enumerate}[1.]
\item Initialize the output sample $S = \emptyset$ and a queue $Q = \emptyset$.
\item Scan the two arrays in $\AFRS$. For each bin $b$, if $BS$ holds a single element, $Enqueue(Q,b)$.
\item While $Q \neq \emptyset$:
\begin{enumerate}[3.1.]
\item $b \leftarrow Dequeue(Q)$. If $BS$ in $b$ holds a single element:
	\begin{enumerate}[3.\mbox{1}.1.]
	\item Extract the element $(k,C_k)$ from $BS$ in $b$.
	\item $S = S \cup \{(k,C_k)\}$.
	\item Subtract $(k,C_k)$ from $BS$ in $\tilde{b}$, where $\tilde{b}$ is the other bin $k$ is hashed to. 
	\item	$Enqueue(Q,\tilde{b})$.
	\end{enumerate}
	\end{enumerate}
\nopagebreak
\item Return $S$.
\end{enumerate}

\begin{lemma} \label{lemma:alg3works}
All elements that do not belong to a fail set are identified by the algorithm.
\end{lemma}

\begin{proof} See Appendix \ref{appendix:proofs}. \end{proof}

\begin{lemma} \label{lemma:alg3time}
The recovery algorithm takes $O(K)$ time.
\end{lemma}

\begin{proof} 
If the algorithm is implemented with hash computations for finding the other bin an element is hashed to, it takes $O(K {(\log\log\frac{K}{\delta})}^2)$ time. Using an additional counter in each $BS$ reduces the time to $O(K)$.
See Appendix \ref{appendix:proofs} for the complete proof. \end{proof}

Let $X$ be the elements in $\AFRS$, $K \leq \left| X \right| \leq \tilde{K}$, $\tilde{K} = \Theta(K)$.
In order to recover all but $\epsilon \left| X \right|$ of the elements we require $K = \Omega(\frac{1}{\epsilon} \log{\frac{1}{\delta}})$.
If $K$ is smaller, we recover all but at most $O(\max{\{\epsilon \left|X \right|, f\}})$ of the elements, where $f=O(\log{\frac{K}{\delta}})$ is the size of the fail set.

\subsubsection{Non-strict $\AFRS$}
In the Non-strict Turnstile model we keep $BS_{ns}$ in each bin, and we set the range of the hash function to $q=\Theta(\frac{K}{\delta})$ and its independence to $t'=\Theta(\log{\frac{K}{\delta}})$, the same as the independence of the hash functions we use when mapping to the bins in $\AFRS$.
The same hash function can be used for all $BS_{ns}$s.
Recall that if $BS_{ns}$ contains a single element, this element is extracted successfully.
If $BS_{ns}$ contains less than $t'$ elements, an event called a \emph{small collision}, the probability of an error is at most $1/q$.
If $BS_{ns}$ contains $t'$ elements or more, an event called a \emph{large collision}, we do not have a guarantee on the probability of an error.

\begin{lemma}
Let $\AFRS$ with at most $\tilde{K}$ elements have 2 arrays of size $s = 4\tilde{K}$.
Let the mapping be performed by two $t$-wise independent hash functions for $t=2\log{\frac{\tilde{K}}{\delta}}$.
Let each bin contain $BS_{ns}$ with $q=\frac{4\tilde{K}}{\delta}$ and $t' = t$.
The probability of no false detections during the entire recovery process is at least $1-\delta$.
\end{lemma}

\begin{proof}
First we bound the probability of a large collision in $\AFRS$. Let $\AFRS$ have $\left|X\right| \leq \tilde{K}$ elements.
Let $\mathcal{E}_{b}$ be the event that there is a large collision in bin $b$.
 Since $t=t'$, every $t'$ elements that appear in a large collision are mapped there independently.
Thus, $Pr[\mathcal{E}_{b}] \leq {X \choose t'} \left(\frac{1}{s}\right)^{t'} 
\leq \left(\frac{eX}{t'}\right)^{t'} (\frac{1}{4\tilde{K}})^{t'}
\leq \left(\frac{e}{4t'}\right)^{t'} < {(\frac{\delta}{\tilde{K}})}^{2}$.
Using the union bound, the probability that there is a large collision in any bin is at most $\delta/2$.
Hence we can consider only small collisions.

The total number of inspections of bins with collisions during the recovery process is at most $2\tilde{K}$. Therefore
the probability of detecting at least one false element as a result of a small collision is at most
$2\tilde{K}\frac{1}{q}=\delta/2$, 
and the probability of any false detections is $\delta$.
\end{proof}

\begin{corollary}
If $BS_{ns}$ with $q=\Theta(\frac{K}{\delta})$ and $t'=\Theta(\log{\frac{K}{\delta}})$ is placed in each bin, the recovery procedure of the Strict Turnstile model can be used also for Non-strict Turnstile data streams with $1-\delta$ success probability.
\end{corollary}

\section{Applications and Extensions} \label{sec:apps}

\paragraph*{Inverse Distribution}

The samples generated by the algorithms $\FRS$ and $\AFRS$ can be used to derive an additive $\epsilon$-approximation with $1-\delta$ success probability for various forward and inverse distribution queries.
For example, consider \emph{Inverse point queries}, which return the value of $f^{-1}(i)$ for a query frequency
$i$. The samples from $\FRS$ and $\AFRS$ can be used to obtain an approximation in $[f^{-1}(i) - \epsilon, f^{-1}(i) + \epsilon]$ for every frequency
$i$.
We can approximate Inverse range queries, Inverse
heavy hitters and Inverse quantiles queries in a similar way.

\begin{lemma} \label{theorem:inverse}
Let $S$ be a $(t,\epsilon')$-partial sample of size $\left| S \right| = \Omega(\frac{1}{\epsilon^{2}}\log\frac{1}{\delta})$ for $\epsilon \in (0,1)$, $\epsilon'=\Theta(\epsilon)$, $\delta \in (0,1)$, and $t = \Omega(\log\frac{1}{\delta})$.
The estimator 
$
{f^{-1}}(i) \approx \frac{\left|\{k \in S \colon C_k = i \} \right|}{\left| S \right|}
$
provides 
an additive $\epsilon$-approximation to the inverse
distribution with probability at least $1-\delta$. 
\end{lemma}

\begin{proof}
See Appendix \ref{appendix:proofs}. \end{proof}

\paragraph*{Union and Difference}

Let $DS_{r,D_{i}}$ be the data structure obtained from data stream $D_{i}$ using the random bits $r$.
The union of streams $D_{1}$, $D_{2}$ is $DS_{r,D_{1}\cup D_{2}}=DS_{r,D_{1}}+DS_{r,D_{2}}$,
where the addition operator adds all $BS$s in all bins of all arrays.
Our sampling algorithm can extract a sample from a union of data streams. This
feature is useful when there are multiple entry points and each of them can update its own data structure locally and then a unified sample can be derived.

Sampling from the difference of streams $DS_{r,D_{1}-D_{2}}=DS_{r,D_{1}}-DS_{r,D_{2}}$ is similar. Note that
even if $D_{1}$ and $D_{2}$ are Strict Turnstile data streams, their difference might represent a Non-strict Turnstile stream. Hence our structures for the
Non-strict Turnstile model are useful for both input streams
in the Non-strict model and for sampling the difference.

\bibliographystyle{plain}
\bibliography{work}

\begin{thebibliography}{10}

\bibitem{charikar04}
Moses Charikar, Kevin Chen, and Martin Farach-Colton.
\newblock Finding frequent items in data streams.
\newblock {\em Theor. Comput. Sci.}, 312(1):3--15, 2004.

\bibitem{cohen12}
Edith Cohen, Graham Cormode, and Nick~G. Duffield.
\newblock Don't let the negatives bring you down: sampling from streams of
  signed updates.
\newblock In {\em SIGMETRICS}, pages 343--354, 2012.

\bibitem{cormode03}
Graham Cormode, Mayur Datar, Piotr Indyk, and S.~Muthukrishnan.
\newblock Comparing data streams using hamming norms (how to zero in).
\newblock {\em IEEE Trans. Knowl. Data Eng.}, 15(3):529--540, 2003.

\bibitem{cormode05}
Graham Cormode, S.~Muthukrishnan, and Irina Rozenbaum.
\newblock Summarizing and mining inverse distributions on data streams via
  dynamic inverse sampling.
\newblock In {\em VLDB}, pages 25--36, 2005.

\bibitem{datar03}
Mayur Datar and S.~Muthukrishnan.
\newblock Estimating rarity and similarity over data stream windows.
\newblock In {\em ESA}, pages 323--334, 2002.

\bibitem{frahling05}
Gereon Frahling, Piotr Indyk, and Christian Sohler.
\newblock Sampling in dynamic data streams and applications.
\newblock In {\em Symposium on Computational Geometry}, pages 142--149, 2005.

\bibitem{frahling08}
Gereon Frahling, Piotr Indyk, and Christian Sohler.
\newblock Sampling in dynamic data streams and applications.
\newblock {\em Int. J. Comput. Geometry Appl.}, 18(1/2):3--28, 2008.

\bibitem{ganuly07}
Sumit Ganguly.
\newblock Counting distinct items over update streams.
\newblock {\em Theor. Comput. Sci.}, 378(3):211--222, 2007.

\bibitem{gemulla06}
Rainer Gemulla, Wolfgang Lehner, and Peter~J. Haas.
\newblock A dip in the reservoir: Maintaining sample synopses of evolving
  datasets.
\newblock In {\em VLDB}, pages 595--606, 2006.

\bibitem{gemulla07}
Rainer Gemulla, Wolfgang Lehner, and Peter~J. Haas.
\newblock Maintaining bernoulli samples over evolving multisets.
\newblock In {\em PODS}, pages 93--102, 2007.

\bibitem{gibbons01}
Phillip~B. Gibbons.
\newblock Distinct sampling for highly-accurate answers to distinct values
  queries and event reports.
\newblock In {\em VLDB}, pages 541--550, 2001.

\bibitem{gilbert10}
Anna~C. Gilbert, Yi~Li, Ely Porat, and Martin~J. Strauss.
\newblock Approximate sparse recovery: optimizing time and measurements.
\newblock In {\em STOC}, pages 475--484, 2010.

\bibitem{indyk01}
Piotr Indyk.
\newblock A small approximately min-wise independent family of hash functions.
\newblock {\em J. Algorithms}, 38(1):84--90, 2001.

\bibitem{jowhari11}
Hossein Jowhari, Mert Saglam, and G{\'a}bor Tardos.
\newblock Tight bounds for lp samplers, finding duplicates in streams, and
  related problems.
\newblock In {\em PODS}, pages 49--58, 2011.

\bibitem{kane10}
Daniel~M. Kane, Jelani Nelson, and David~P. Woodruff.
\newblock An optimal algorithm for the distinct elements problem.
\newblock In {\em PODS}, pages 41--52, 2010.

\bibitem{karamcheti05}
Vijay Karamcheti, Davi Geiger, Zvi~M. Kedem, and S.~Muthukrishnan.
\newblock Detecting malicious network traffic using inverse distributions of
  packet contents.
\newblock In {\em MineNet}, pages 165--170, 2005.

\bibitem{micciancio97}
Daniele Micciancio.
\newblock Oblivious data structures: Applications to cryptography.
\newblock In {\em STOC}, pages 456--464, 1997.

\bibitem{mone10}
Morteza Monemizadeh and David~P. Woodruff.
\newblock 1-pass relative-error l$_{\mbox{p}}$-sampling with applications.
\newblock In {\em SODA}, pages 1143--1160, 2010.

\bibitem{muth05}
S.~Muthukrishnan.
\newblock Data streams: Algorithms and applications.
\newblock {\em Foundations and Trends in Theoretical Computer Science}, 1(2),
  2005.

\bibitem{naor01}
Moni Naor and Vanessa Teague.
\newblock Anti-persistence: History independent data structures.
\newblock {\em IACR Cryptology ePrint Archive}, 2001:36, 2001.

\bibitem{pagh08}
Anna Pagh and Rasmus Pagh.
\newblock Uniform hashing in constant time and optimal space.
\newblock {\em SIAM J. Comput.}, 38(1):85--96, 2008.

\bibitem{pagh01}
Rasmus Pagh and Flemming~Friche Rodler.
\newblock Cuckoo hashing.
\newblock In {\em ESA}, pages 121--133, 2001.

\bibitem{porat07}
Ely Porat and Ohad Lipsky.
\newblock Improved sketching of hamming distance with error correcting.
\newblock In {\em CPM}, pages 173--182, 2007.

\bibitem{tao07}
Yufei Tao, Xiang Lian, Dimitris Papadias, and Marios Hadjieleftheriou.
\newblock Random sampling for continuous streams with arbitrary updates.
\newblock {\em IEEE Trans. Knowl. Data Eng.}, 19(1):96--110, 2007.

\bibitem{gathen99}
Joachim von~zur Gathen and J\"{u}rgen Gerhard.
\newblock {\em Modern computer algebra}.
\newblock Cambridge University Press, New York, NY, USA, 1999.

\end{thebibliography}

\newpage

\appendix

\section{Approximation Using High Moments}

\label{appendix:chebyshev}

For a random variable $Z$ and an even number $l$,\[
\Pr[\left|Z-E[Z]\right|\geq t]=\Pr[(Z-E[Z])^{l}\geq t^{l}]\leq\frac{\Delta^{l}}{t^{l}}\]
where $\Delta^{l}=E[(Z-E[Z])^{l}]$ is the $l$'th central moment
of $Z$.
We use the estimation of \cite{indyk01} to $\Delta^{l}$, where $Z$ is a sum of $l$-wise independent indicator variables:
$\Delta^{l}\leq8(6l)^{(l+1)/2}{E[Z]^{(l+1)/2}}$.

This implies:\[
\Pr[\left|Z-E[Z]\right|\geq\alpha E[Z]]\leq\frac{48l}{\alpha}{\left(\frac{6l}{\alpha^{2}E[Z]}\right)}^{(l-1)/2}\label{eq:l_cheby}\]

\begin{lemma}\label{lemma:l_cheby} Let $Z$ be a sum of $l$-wise independent indicator
variables with $E[Z]=\frac{c}{\epsilon^{2}}\log\frac{1}{\delta}$
for some $0<\epsilon$, $\delta \in (0,1)$ and constant $c$. Let $l=\tilde{c} \log\frac{1}{\delta}$ for some constant $\tilde{c}$
be an even number.
Then for big enough constants $c, \tilde{c}$:
$\Pr[\left|Z-E[Z]\right|\geq\epsilon E[Z]]<\delta$.
 \end{lemma}

\begin{proof} \begin{align*}
\Pr[\left|Z-E[Z]\right|\geq\epsilon E[Z]] & 
\leq\frac{48 \tilde{c} \log\frac{1}{\delta}}{\epsilon}
{\left(\frac{6 \tilde{c}\log\frac{1}{\delta}}{\epsilon^{2}\frac{c}{\epsilon^{2}}\log\frac{1}{\delta}}\right)}
^{(\tilde{c}\log\frac{1}{\delta}-1)/2}\\
 & =\frac{48\tilde{c}\log\frac{1}{\delta}}{\epsilon}{\left(\frac{6\tilde{c}}{c}\right)}^{(\tilde{c}\log\frac{1}{\delta}-1)/2}<\delta\end{align*}
For $\epsilon > \poly{(\delta)}$.
\end{proof}

Note that for a constant $\alpha$, in order to prove $\Pr[\left|Z-E[Z]\right|\geq\alpha E[Z]]<\delta$, $E[Z] = \Omega{(\log{\frac{1}{\delta}})}$ is sufficient.

\section{Proofs from Section \ref{sec:work}} \label{appendix:proofs}

\subsection{Proof of Lemma \ref{lemma:mapping}}

\begin{proof} Let $X_{l}$ be a random variable that indicates the
number of elements $k$ in the stream for which $(h_{l}(k)=0\wedge
h_{l}(k)\neq0)$. I.e. $X_{l}$ is the number of elements in level
$l$. $\mathcal{H}$ is a $t$-wise independent
hash family for $t = \Theta(\log{\frac{1}{\delta}})$.
Therefore for each $l\in L$ and $k$ in the stream, $\Pr_{h\in \mathcal{H}}[h_{l}(k)=0\wedge h_{l+1}(k)\neq0]=\lambda^{l}-\lambda^{l+1}=\lambda^{l}(1-\lambda)$.
We denote $p_{l}=\lambda^{l}(1-\lambda)$.

$L_{0}\leq\tilde{L}_{0}\leq\alpha L_{0}$ implies $\frac{1}{\alpha}\tilde{L}_{0}\leq L_{0}\leq\tilde{L}_{0}$
and we obtain:
$
{\frac{1}{\alpha}\tilde{L}_{0}p_{l^{*}+1}}<2K\leq{\frac{1}{\alpha}\tilde{L}_{0}p_{l^{*}}}\leq{L_{0}p_{l^{*}}}\leq{\tilde{L}_{0}p_{l^{*}}}\leq\frac{2\alpha K}{\lambda}
$.
The expected number of elements in level $l$ is $E[X_{l}]=L_{0}p_{l}$. Hence for level $l^{*}$, $
2K\leq E[X_{l^{*}}]\leq\frac{2\alpha K}{\lambda}$.

We write $E[X_{l^{*}}]=\beta K$ for some $2\leq\beta\leq\frac{2\alpha}{\lambda}$.
From Lemma \ref{lemma:l_cheby} we get: $
\Pr[\left|X_{l^{*}}-E[X_{l^{*}}]\right|\geq K]=\Pr[\left|X_{l^{*}}-E[X_{l^{*}}]\right|\geq\frac{1}{\beta}E[X_{l^{*}}]]<\delta $.
We can use the lemma since $X_{l^{*}}$ is a sum of $t$-wise independent variables, $t = \Theta(\log\frac{1}{\delta})$, and $K=\Omega(\log\frac{1}{\delta})$.
\end{proof}

\subsection{Proof of Lemma \ref{Collisions in arrays}}

\begin{proof} 
First we bound the probability that a specific element $k$ collides with another element in a specific array $a$.
Let $\mathcal{\mathcal{C}}_{kj}^{a}$ be the event
that elements $k$ and $j$ collide in array $a$.
Since pairwise independent hash functions are used to map the elements to bins, $\forall k,j,a\quad\Pr[\mathcal{\mathcal{C}}_{kj}^{a}]=\frac{1}{s}=\frac{1}{2\tilde{K}}$.
Let $\mathcal{C}_{k}^{a}$ be the event that $k$ collides with any
element in array $a$.
$\Pr\left[\mathcal{C}_{k}^{a}\right]=\Pr[\exists{j}\quad\mathcal{C}_{kj}^{a}]\leq\bigcup_{j\neq k}\ \Pr[\mathcal{\mathcal{C}}_{kj}^{a}]<\tilde{K} \cdot\Pr[\mathcal{\mathcal{C}}_{kj}^{a}]=\frac{1}{2}$.

Now we prove that with probability $1-\delta/{\tilde{K}}$ there is an array in which no element collides with $k$.
Let $\left| \mathcal{C}_{k}\right|$ be the number of arrays in which $k$ collides with another element.
We want to show $\left| {\mathcal{C}}_{k} \right| < \tau$ with high probability.
Let $X$ be the set of elements in $\FRS$. We know $\left| X \right| \leq \tilde{K}$. The hash functions of the different arrays are independent.
Therefore 
$\Pr[\left|\mathcal{\mathcal{C}}_{k}\right|=\tau]=\Pr[\forall a\in[\tau], \; \mathcal{C}_{k}^{a}]
\le\Pr[\forall a\in[\tau], \; \mathcal{C}_{k}^{a} \mid \left|X\right|=\tilde{K}]
=\prod_{a\in[\tau]}{\Pr[\mathcal{C}_{k}^{a} \mid \left|X\right|=\tilde{K}]}<\left(\frac{1}{2}\right)^{\tau}=\frac{\delta}{\tilde{K}}$.

We conclude with: $\Pr[\exists k, \; \mbox{$k$ collides in all arrays}]=\Pr[\exists k, \; \left|\mathcal{\mathcal{C}}_{k}\right|=\tau]$
$ \leq \bigcup_{k}\Pr[\left|\mathcal{\mathcal{C}}_{k}\right|=\tau]<\delta$.
\end{proof}

\subsection{Proof of Lemma \ref{lemma:alg3works}}

\begin{proof} If the set of unidentified elements is not a fail
set, then one of them is isolated in a bin.
Let $x$ be the element and $b$ be the bin in which $x$ is the single element.
If $x$ was isolated prior to the first scan over all bins, bin $b$ would have been identified in the first scan.
If $x$ became isolated during the recovery process, then all other elements in $b$ were identified beforehand.
When the last of them was identified and removed from $b$, $x$ became isolated, and $b$ was inserted to the queue. 
Each bin inserted to the queue is removed from the queue before the process ends, and hence when $b$ was removed, $x$ was identified.
Identifying $x$ results in a smaller set of unidentified elements.
The identification process can proceed until all elements are identified or none
of them are isolated, i.e. they all belong to a fail set. \end{proof}

\subsection{Proof of Lemma \ref{lemma:alg3time}}

\begin{proof}
The number of operations in the initial scan is linear in the number of bins in $\AFRS$, which is $O(K)$. The number of bins
inserted to $Q$ in the process is $O(K)$, because a bin is inserted only when an element is extracted and added to the sample.
Thus, apart from the operation of identifying the other bin an element is hashed to, all other operations take a total of $O(K)$ time.

Assume the phase of finding the other bin an extracted element $k$ is hashed to, is implemented in the algorithm by evaluating the hash function on $k$. 
This evaluation occurs $O(K)$ times. The hash functions are $t$-wise independent, where $t=\Theta(\log\frac{K}{\delta})$.
Thus, in this implementation the recovery time for all $O(K)$ elements is $O(K ({\log\log{\frac{K}{\delta}}})^2)$.

The recovery time can be reduced to $O(K)$ by using an additional counter $W$ in each $BS$.
Let $h_1$, $h_2$ be the two hash functions 
that map the elements to the bins in the two arrays.
When inserting an element $(x_i,c_i)$ to $BS$ in array $a \in {\{1,2}\}$, the following update is performed:
 $W \leftarrow W + c_i h_{3-a}(x_i)$.
The space and update time required by the algorithm remain of the same order, 
since the update takes an additional $O(1)$ time and the space of $W$ is $O(\log{(mr)})$ bits.

If $BS$ in array $a$ contains a single element $(k,C_k)$, then $W = C_k h_{3-a}(k)$. 
Thus, if $(k,C_k)$ is extracted from $BS$ in array $a$ we obtain its location $h_{3-a}(k)$ in the other array $3-a$ without evaluating a hash function. 
The recovery time is $O(K)$ for all $O(K)$ elements.
\end{proof}

\subsection{Proof of Lemma \ref{theorem:inverse}}

\begin{proof}

Our estimator is $
f^{-1}(i) \approx \frac{\left|\{k \in S \colon C_k = i \} \right|}{\left| S \right|}
$.
We need to prove that it is an $\epsilon$-approximation to $f^{-1}(i) = \frac{\left|\{k \colon C_k = i \} \right|}{\left|\{k \colon C_k \neq 0 \} \right|}$, the fraction of distinct values with total count equals $i$, with probability $1-\delta$.

We prove that it is an $\epsilon$-approximation when all elements are recovered from the recovery structure, i.e. when $\FRS$ is used. Later we relax the assumption that all elements are recovered, and thus prove that we get an $\epsilon$-approximation also when using $\AFRS$. Thus, an $\epsilon$-approximation is provided with $1-\delta$ success probability when using both our sampling algorithms.

The elements recovered are from a specific level $l^{*}$ that depends on $L_{0}$. 
For value $k$, $C_{k}\neq0$, we define a random variable
$Y_{k}$. $Y_{k}=1$ if $(k,C_{k})$ is mapped to level $l^{*}$.
$\Pr[Y_{k}=1]=\lambda^{l^{*}}(1-\lambda)$, and we denote $p_{l^{*}}=\lambda^{l^{*}}(1-\lambda)$.

Let $Y=\sum_{k}{Y_{k}}$ be the number of elements in level $l^{*}$.
For now assume that all elements in the level are recovered, i.e.
$Y=\left|S\right|$. Later we relax this assumption. $E[Y]$ is the expected sample
size obtained from level $l^{*}$. Thus $E[\left|S\right|]=E[Y]=L_{0}p_{l^{*}}$.
We choose $l^{*}$ as a level with $\Theta(K)$ elements, i.e. $E[Y] =\Theta(K) = \Omega(\frac{1}{\epsilon^2}\log\frac{1}{\delta})$. 

${\{Y_{k}\}}$ are $t$-wise independent, where $t = \Theta(\log\frac{1}{\delta})$. $E[Y] = \Omega(\frac{1}{\epsilon^2}\log\frac{1}{\delta})$. From Lemma~\ref{lemma:l_cheby}
we obtain $\Pr[\left|Y-E[Y]\right|>\epsilon'E[Y]]<\delta'$, for $\epsilon'=\Theta(\epsilon)$ and $\delta' = \Theta(\delta)$
that will be determined later.
Thus, with probability $1-\delta'$, \[
(1-\epsilon')E[\left|S\right|]\leq\left|S\right|\leq(1+\epsilon')E[\left|S\right|] \]

Let $F_{i}=\sum_{k \colon C_{k}=i}{Y_{k}}$, be the number of elements
in $S$ with frequency $i$. $
E[F_{i}]=\left|\{k \colon C_{k}=i\}\right|\cdot p_{l}^{*}=f^{-1}(i)\cdot E[\left|S\right|] $.
We get $f^{-1}{(i)}=\frac{E[F_{i}]}{E[\left|S\right|]}$. If
all elements in the level are recovered, the estimator of $f^{-1}(i)$ can be written as:
$f^{-1}{(i)}\approx\frac{F_{i}}{\left|S\right|}$.

Hence we would like to prove: \begin{equation}
\Pr\left[\left|f^{-1}(i)-\frac{F_{i}}{\left|S\right|}\right|\geq\epsilon\right]\leq\delta\label{eq:estimator}\end{equation}

$\frac{E[F_{i}]}{\left|S\right|}$ is an $\epsilon'$-approximation to $\frac{E[F_{i}]}{E[\left|S\right|]}$, i.e. $\Pr \left[\left|f^{-1}(i) - \frac{E[F_{i}]}{\left|S\right|} \right|\geq\epsilon' \right]\leq\delta'$ since
\[
\frac{E[F_{i}]}{E[\left|S\right|]}-\epsilon'\leq\frac{E[F_{i}]}{(1+\epsilon')E[\left|S\right|]}\leq\frac{E[F_{i}]}{\left|S\right|}\leq\frac{E[F_{i}]}{(1-\epsilon')E[\left|S\right|]}\leq\frac{E[F_{i}]}{E[\left|S\right|]}+\epsilon'\]

with probability $1-\delta'$ when $f^{-1}(i)=\frac{E[F_{i}]}{E[\left|S\right|]}\leq1-\epsilon'$.
Then by using Lemma \ref{lemma:l_cheby} again we get (\ref{eq:estimator}),
for a constant $c$, $2\epsilon'$-additive error and $1-2\delta'$ success probability.

Now we relax the assumption that all elements in the level were recovered.
The maximal bias occurs if $\epsilon'\left|S\right|$ elements were not recovered,
and all unrecovered elements had frequency $i$. I.e. $\left|\{k\in S \colon C_{k}=i\}\right|=F_{i}\pm\epsilon'\left|S\right|$.
Thus, there is an additional additive error of $\epsilon'$. Setting
$\epsilon'=\epsilon/3$, $\delta' = \delta/3$ completes the proof.
\end{proof}

\end{document}